\documentclass[11pt,letterpaper]{article}
\pdfoutput=1

\usepackage{graphicx}
\usepackage{amsmath,amsthm,amssymb,amsfonts}
\usepackage{authblk}
\usepackage{thm-restate}
\newcommand{\Oh}{\mathcal{O}}
\newcommand{\opsquare}{op-square}
\newcommand{\op}{\approx}

\newcommand{\iprev}[1]{\textsf{prev}_{#1}}
\newcommand{\ikstart}[2]{\textsf{begin}_{#1,#2}}
\newcommand{\ikend}[2]{\textsf{end}_{#1,#2}}
\DeclareMathOperator{\code}{code}
\DeclareMathOperator{\prev}{prev}
\DeclareMathOperator{\cnt}{cnt}
\DeclareMathOperator{\fingerprint}{fingerprint}
\usepackage[noadjust]{cite}
\usepackage{todonotes}
\usepackage[shortlabels]{enumitem}

\usepackage[left=1in,right=1in,top=1in,bottom=1in]{geometry}

\newtheorem{theorem}{Theorem}[section]
\newtheorem{lemma}[theorem]{Lemma}

\newtheorem{definition}[lemma]{Definition}

\newtheorem{proposition}[lemma]{Proposition}

\makeatletter
\newcommand\thankssymb[1]{\textsuperscript{\@fnsymbol{#1}}}
\makeatother

\title{Order-Preserving Squares in Strings}
\date{}

\author[1]{Pawe\l{} Gawrychowski}
\author[2]{Samah Ghazawi\thanks{Partially supported by the Israel Science Foundation grant 1475/18, and Grant No. 2018141 from the United States-Israel
Binational Science Foundation (BSF).}}
\author[2,3]{Gad M. Landau\thankssymb{1}}

\affil[1]{Institute of Computer Science, University of Wroc\l{}aw, Poland}
\affil[2]{Department of Computer Science, University of Haifa, Israel}
\affil[3]{NYU Tandon School of Engineering, New York University, Brooklyn, NY, USA}

\begin{document}
\maketitle
\begin{abstract}
An order-preserving square in a string is a fragment of the form $uv$ where $u\neq v$ and $u$ is order-isomorphic to $v$.
We show that a string $w$ of length $n$ over an alphabet of size $\sigma$ contains $\Oh(\sigma n)$ order-preserving squares
that are distinct as words. This improves the upper bound of $\Oh(\sigma^{2}n)$ by Kociumaka, Radoszewski, Rytter, and Waleń
[TCS 2016]. Further, for every $\sigma$ and $n$ we exhibit a string with $\Omega(\sigma n)$ order-preserving squares
that are distinct as words, thus establishing that our upper bound is asymptotically tight. Finally, we design
an $\Oh(\sigma n)$ time algorithm that outputs all order-preserving squares that occur in a given string and
are distinct as words. By our lower bound, this is optimal in the worst case.
\end{abstract}

\section{Introduction}
A natural definition of repetitions in strings is that of squares, which are fragments of the form $uu$, where $u$ is a string.
The study of repetitions in strings goes back at least to the work of Thue from 1906~\cite{thue}, who constructed an infinite
square-free word over the ternary alphabet. Since then, multiple definitions of repetitions have been proposed
and studied, with the basic question being focused on analyzing how many such repetitions a string of length $n$ can contain.
Of course, any even-length fragment of the string $\texttt{a}^{n}$ is a square, therefore we would like to count distinct
squares.
Using a combinatorial result of Crochemore and Rytter~\cite{Crochemore95}, Fraenkel and Simpson~\cite{FS} proved that a string of length $n$ contains at most $2n$ distinct squares (also see a simpler proof by Ilie~\cite{Ilie05}). They also provided an infinite family of strings of length $n$ with $n-o(n)$ distinct squares.
For many years, it was conjectured that the right upper bound is actually $n$. Interestingly, a proof of the conjecture for the binary alphabet would imply it for any alphabet~\cite{ManeaS15}. 
Very recently, after a series of improvements on the upper bound~\cite{Ilie07,Lam,DEZA,Thierry2020},
the conjecture has been finally resolved by Brlek and Li~\cite{Brlek}, who showed an upper bound of $n-\sigma+1$,
where $\sigma$ is the size of the alphabet.

For many of the applications, it seems more appropriate to work with different definitions of equality, giving us different notions of squares.
Three interesting examples are (1) Abelian squares~\cite{CS,Kociumaka,erdos,AAE,HUOVA,Asquare,Kociumaka2,pleasants} (also called Jumbled squares) that are of interest in natural language processing applications and in other domains where the classifications strongly depend on feature sets distribution, as opposed to feature sequences distributions. (2) Parameterized squares~\cite{Kociumaka} that are considered in applications for finding identical sections of code. (3) Order-preserving squares~\cite{Kociumaka,CROCHEMORE2016122,GOURDEL2020} that are important for applications of stock price analysis and musical melody matching.

The combinatorial properties of the three types of squares were studied by Kociumaka et al.~\cite{Kociumaka}. Given a string of length $n$ over an alphabet of size $\sigma$, first the authors bounded the number of distinct as words abelian squares by $\Theta(n^2)$. Second, bounded the number of distinct as words parameterized squares by $\Oh(2(\sigma!)^2n)$ and bounded the number of nonequivalent parameterized squares (see definition within) by $\Oh(2\sigma!n)$.
Third, the authors provided $\Oh(\sigma^2n)$ bound for the number of distinct as words order-preserving squares.

From an algorithmic perspective, various algorithms were proposed for computing abelian squares and order-preserving squares in a string of length $n$.
Cummings and Smyth~\cite{CS} proposed an $\Theta(n^2)$ time algorithm for computing all substrings that consist of a concatenation of two or more abelian-equivalent substrings. Kociumaka et al.~\cite{Kociumaka2} proposed an algorithm for computing the longest, the shortest, and the number of all abelian squares in $\Oh(n^2/\log^2n)$ time using linear space.
 Gourdel et al.~\cite{GOURDEL2020} proved that all nonshiftable order-preserving squares (see definition within) can be computed in $\Oh(n \log n)$ time. Additionally, Crochemore et al.~\cite{CROCHEMORE2016122} proposed the \textit{incomplete} order-preserving suffix tree (see details within), denoted by $T$, that enables order-preserving pattern matching queries in time proportional to the pattern length. $T$ can be constructed in $\Oh(n\log\log n)$ expected time and $\Oh(n\log^2\log n/\log\log\log n)$ worst-case time. Moreover, the authors proved that using $T$, all occurrences of order-preserving squares can be computed in $\Oh(n\log n+occ)$ time, where $occ$ is the total number of occurrences of order-preserving squares. Note that, the number of all occurrences of order-preserving squares might be unreasonably high.
In particular, every regular square is considered to be an order-preserving square, hence $\texttt{a}^{n}$ contains $\Theta(n^{2})$
occurrences of order-preserving squares.
Henceforth, a more natural approach is to generate only order-preserving squares that are distinct as words. 
\subparagraph*{Our results.}
In this paper, we focus on order-preserving squares. Same-length strings $u$ and $v$ over an ordered alphabet are order-isomorphic,
denoted $u\op v$, when the order between the characters at the corresponding positions is the same in $u$ and $v$. 
For example, $u=acb$ and $v=azd$ are order-isomorphic, assuming $a<b<c<d<z$. In this paper, order-preserving squares are
strings of the form $uv$, where $u\op v$ and additionally $u\neq v$.

The main result of our paper is that the number of order-preserving squares in a string of length $n$ over an alphabet $\sigma$ is $\Oh(\sigma n)$.
This improves the bound of $\Oh(\sigma^{2}n)$ by Kociumaka et al.~\cite{Kociumaka}.
We stress that in our definition of an order-preserving square, we require that $u\neq v$, while Kociumaka et al.~\cite{Kociumaka} counted
fragments of the form $uv$, where $u\op v$, that are distinct as words. We believe that our definition is more natural in the context
of this paper.
At the same time, by the result of Brlek and Li~\cite{Brlek} a string of length $n$ contains
less than $n$ fragments of the form $uu$ that are distinct as words, thus our result implies that the number of fragments $uv$
such that $u\op v$ that are distinct as words is also $\Oh(\sigma n)$.
We complement our upper bound by designing, for each $\sigma$, an infinite family of strings of length $n$ over an alphabet of size $\sigma$
containing $\Theta(\sigma n)$ such fragments.
We begin with describing the lower bound in Section~\ref{lowerop}, and then present the upper bound in Section~\ref{upperop}.

\begin{theorem}
\label{bound}
The number of order-preserving squares in a string of length $n$ over an alphabet of size $\sigma$ is $\Oh(\sigma n)$,
and this bound is asymptotically tight even if we only consider order-preserving squares that are distinct as words.
\end{theorem}

Next, we design an algorithm for reporting all order-preserving squares in a given string of length $n$ over an alphabet $\sigma$
in $\Oh(\sigma n)$ time, which (by our lower bound) is asymptotically optimal in the worst case. 
We again stress that in our definition of an order-preserving square, we require that $u\neq v$.
However, all fragments of the form $uu$ that are distinct as words can be reported in $\Oh(n)$ by the algorithm of Gusfield and Stoye~\cite{GUSFIELD},
Thus, for $\sigma = o(\log n)$, this resolves one of the open questions by Crochemore et al.~\cite{CROCHEMORE2016122}, who asked
if there is an $o(n\log n)$ time algorithm for finding the longest order-preserving square.
This is described in Section~\ref{algo}.

\begin{restatable}{theorem}{reporting}
\label{thm:reporting}
All order-preserving squares in a string of length $n$ over an alphabet of size $\sigma$ can be found in $\Oh(\sigma n)$ time.
\end{restatable}

\subparagraph*{High-level description of our techniques.}
For the lower bound, first, we consider the increasing string $w=123\ldots n$ where $\sigma=n$. Clearly, any even-length fragment is
an order-preserving square thus producing the maximum number, i.e. $\Omega(n^2)=\Omega(\sigma n)$, of order-preserving squares
in a string of length $n$. To decrease the size of the alphabet $\sigma$, we replace $w$ with a non-decreasing string
$w=11\ldots122\ldots2\ldots\sigma\sigma\ldots\sigma$, where each character is repeated the same number of times.
We exhibit $\Omega(\sigma n)$ order-preserving squares in $w$ that are distinct as words.
See Section~\ref{lowerop} for more details.

For the upper bound, we build on the insight by Kociumaka et al.~\cite{Kociumaka}, where the high-level strategy is to consider each
suffix of $w$ separately. For each suffix, they considered the set of leftmost occurrences, consisting of the first occurrence
of each character of the alphabet. Thus, there are at most $\sigma$ leftmost occurrences in each suffix.
For a fixed suffix, they considered all of its prefixes as possible order-preserving squares $uv$. Next,
They showed that, because $u\neq v$, $uv$ is defined by a pair of leftmost occurrences such that one occurrence belongs to $u$,
and the other one belongs to $v$ at the same relative position. For example, let $acbadbxz$ be the suffix, then the pair of leftmost
occurrences $2$ and $5$ defines the order-preserving square $acbadb$.
Thus, as a result, they upper bounded the number of order-preserving squares being a prefix of the considered suffix by ${\sigma \choose 2}$, so ${\sigma \choose 2}n$ in total.

We also separately upper bound the number of order-preserving squares that are prefixes of a suffix of the input string $w$.
However, our goal is to show that there are only $\Oh(\sigma)$ such prefixes, so $\Oh(\sigma n)$ in total.
To this end, we first partition the order-preserving squares into groups. Let $O_{k}$ consists of all order-preserving squares
$uv$ such that $2^{k}\leq |uv| < 2^{k+1}$. Similarly, we partition the leftmost occurrences into groups. Let
$L_{k}$ consists of all leftmost occurrences $i$ such that $2^{k}\leq i < 2^{k+1}$.
Now, our strategy is to show that if $|O_{k}|$ is larger than some fixed constant then $|O_{k}| = \Oh(|L_{k-2}|)$.
The structure of the argument is as follows. We first observe that two order-preserving squares $uv$ and $u'v'$
imply that $|u|-\Delta$, where $\Delta=|u'|-|u|$, is a so-called order-preserving border of $u$.
We write $u=b_{1}b_{2}\ldots b_{f}b_{f+1}$, where $|b_{1}|=|b_{2}|=\ldots |b_{f}|=\Delta$ and $|b_{f+1}|<\Delta$,
and by carefully choosing $uv$ and $u'v'$ from $O_{k}$ conclude that $b_{2}$ contains a leftmost occurrence
and $f$ is proportional to $|O_{k}|$. Then, we argue that $b_{2}$ containing a leftmost occurrence implies
that, in fact, every $b_{j}$ contains a leftmost occurrence, and thus $|O_{k}| = \Oh(|L_{k-2}|)$.
Summing this over all $k$, and separately considering all $k$ such that $|O_{k}|$ is less than the fixed constant,
we are able to conclude that $\sum_{k}|O_k| = \sum_{k}\Oh(|L_k|)<\Oh(\sigma)$.
See Section~\ref{upperop} for more details.

To obtain an efficient algorithm for reporting all order-preserving squares, we apply the order-preserving suffix tree
as defined by Crochemore et al.~\cite{CROCHEMORE2016122}. This structure allows us to check if $w[i..i+2\ell-1]$ is
an order-preserving square by checking if the LCA of two leaves is at string depth at least $\ell$. First, we need
to show how to construct the order-preserving tree in $\Oh(\sigma n)$ time.
Second, we extend the above reasoning to efficiently generate only $\Oh(\sigma n)$ fragments that are then
tested for being an order-preserving square in constant time each.
While the underlying argument is essentially the same as when bounding the number of order-preserving squares,
it needs to be executed differently for the purpose of an efficient implementation.
See Section~\ref{algo} for more details.


\section{Preliminaries}
Let $\Sigma=\{1,\ldots,\sigma\}$ be a fixed finite alphabet of size $\sigma$.
Let $|s|$ denote the length of a string $s$.
For a string $s$, $s[i]$ denotes the character at position $i$ of $s$, and $s[i..j]$ is the fragment of $s$ starting at
position $i$ and ending at position $j$.
We call two strings $u$ and $v$ order-isomorphic, denoted by $u\op v$, when $|u|=|v|$ and, for each $i,j$, we have
$u[i] \leq u[j]$ if and only if $v[i] \leq v[j]$.
The concatenation of two strings $u$ and $v$ is denoted by $uv$.
A string of the form $uv$ is called an order-preserving square, or {\opsquare}, when $u\neq v$ and $u\op v$.
We call $u$ its left arm and $v$ its right arm.
We stress that a regular square, that is, a string of the form $xx$, is not an {\opsquare}.
Two {\opsquare}s $uv$ and $u'v'$ are \textit{distinct as words} if and only if $uv\neq u'v'$.

A trie is a rooted trie, with every edge labeled with a single character and edges outgoing from the same node
having distinct labels. A node $u$ of a trie represents the string obtained by reading the labels on the path from
the root to $u$. A compacted trie is obtained from a trie by replacing maximal paths consisting of nodes with
exactly one child with single edges labeled by the concatenation of the labels of the edges on the path.
A suffix tree $ST$ of a string $w$ is a compacted trie whose leaves correspond to the suffixes of $w\$$.
The string depth of a node $u$ of $ST$ is the length of the string that it corresponds to.
An explicit node of $ST$ is simply a node of $ST$.
An implicit node of $ST$ is a node of the non-compacted trie corresponding to $ST$, or in other words
a location on an edge of $ST$.

Next, we need some definitions specific to order-isomorphism.
Following Kubica et al.~\cite{KUBICA}, we call $b$ an op-border of a string $s[1..n]$ when $s[1..b] \op s[n-b+1..n]$.
Following Gourdel et al.~\cite{GOURDEL2020} (and Matsuoka et al.~\cite{MatsuokaAIBT16}),
we call $p$ an initial op-period of $s[1..n]$  when $s=b_{1}b_{2}\ldots b_{f }b_{f+1}$ with
$|b_{1}|=|b_{2}|=\ldots = |b_{f}|=p$ and $|b_{f+1}|<p$ (so $f=\lfloor n/p\rfloor$),
$b_{1}\op b_{2}\op \ldots \op b_{f}$ and $b_{1}[1..|b_{f+1}|]\op b_{2}[1..|b_{f+1}|] \op \ldots \op b_{f}[1..|b_{f+1}|] \op b_{f+1}$.
$b_{1},b_{2},\ldots,b_{f}$ are called the \textit{blocks} defined by $p$ in $s$, while $b_{f+1}$ (possibly empty) is called the \textit{incomplete block}.
While in the classical setting $p$ is a period of $s[1..n]$ if and only if $n-b$ is a border of $s[1..n]$,
in the order-preserving setting, we only have an implication in one direction.

\begin{proposition}
\label{borderperiods}
If $b$ is an op-border of $s[1..n]$ then $n-b$ is an initial op-period of $s[1..n]$.
\end{proposition}

\begin{proof}
Let $p=n-b$ and $f=\lfloor n/p\rfloor$. We represent $s[1..n]$ as $s=b_{1}b_{2}\ldots b_{f}b_{f+1}$ with
$|b_{1}|=|b_{2}|=\ldots = |b_{f}|=p$ and $|b_{f+1}|<p$. By $b$ being an op-border of $s[1..n]$,
we have $s[1..b]\op s[n-b+1..n]$, so $s[1..n-p]\op s[p+1..n]$. We observe that
$s[1..n-p]=b_{1}b_{2}\ldots b_{f-1}b_{f}[1..|b_{f+1}|]$ and $s[p+1..n]=b_{2}b_{3}\ldots b_{f} b_{f+1}$.
Then, $b_{1}b_{2}\ldots b_{f-1}b_{f}[1..|b_{f+1}|] \op b_{2}b_{3}\ldots b_{f} b_{f+1}$
implies $b_{i} \op b_{i+1}$, for every $i=1,2,\ldots,f-1$, and $b_{i}[1..|b_{f+1}|]\op b_{i+1}[1..|b_{f+1}|]$,
for every $i=1,2,\ldots,f$.
Hence, $b_{1}\op b_{2}\op b_{3} \op \ldots \op b_{f-1} \op b_{f}$ and $b_{1}[1..|b_{f+1}|]\op b_{2}[1..|b_{f+1}|] \op \ldots \op b_{f}[1..|b_{f+1}|] \op b_{f+1}$, so $p=n-b$ is indeed an initial op-period of $s[1..n]$.
\end{proof}
Due to Proposition~\ref{borderperiods}, $b$ being an op-border of $s[1..n]$ implies
that $s[1..n]=b_{1}b_{2}\ldots b_{f} b_{f+1}$, where $b_{1}b_{2}\ldots b_{f}[1..|b_{f+1}|]\op b_{2}b_{3}\ldots b_{f} b_{f+1}$,
$|b_{1}|=|b_{2}|=\ldots = |b_{f}|=n-b$ and $|b_{f+1}|<p$ (so $f=\lfloor n/(n-b)\rfloor$),
$b_{1}\op b_{2}\op \ldots \op b_{f}$ and $b_{1}[1..|b_{f+1}|]\op b_{2}[1..|b_{f+1}|] \op \ldots \op b_{f}[1..|b_{f+1}|] \op b_{f+1}$.
We will say that these blocks are defined by $b$.


\section{Lower Bound}
\label{lowerop}
Recall that $\Sigma=\{1,\ldots,\sigma\}$.
We define a string $w=11\ldots122\ldots2\ldots\sigma\sigma\ldots\sigma$,
that is, a concatenation of $\sigma$ blocks, each consisting of $k$ repetitions of the same character.
We note that $|w|=\sigma k$.
For $i=1,2,\ldots,\lfloor\sigma/2\rfloor $, we consider all fragments of $w$ of length $2ik$
starting at positions $j=1,2,\ldots,|s|-2ik+1$.
For $j=1 \bmod k$, the fragment is a concatenation of $2i$ blocks, each block consisting of $k$ repetitions of the same character.
For $j\neq 1 \bmod k$, the fragment starts with $r\in [1,k-1]$ repetitions of the same character,
then $2i-1$ blocks, each block consisting of $k$ repetitions of the same character,
and finally $\ell=k-r$ repetitions of the same character. See Figure~\ref{fig:ex1}.

\begin{figure}[ht]
\centerline{\includegraphics[scale=.5]{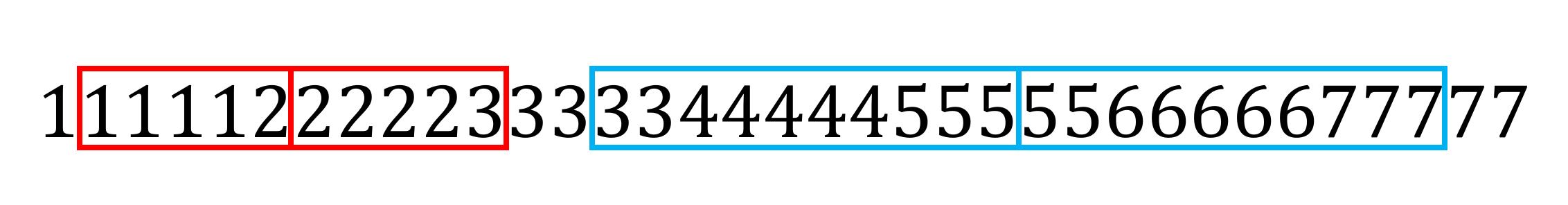}}
\caption{The red box corresponds to an {\opsquare} of length $10$ containing $3$ different characters. The blue box corresponds to an {\opsquare} of length $20$ containing $5$ different characters.}
\label{fig:ex1}
\end{figure}

\noindent Each such fragment is an {\opsquare}. For $j=1 \bmod k$, both the left and the right arm consist of $i$ blocks
consisting of $k$ repetitions of character $c,c+1,\ldots,c+i-1$.
For $j\neq 1 \bmod k$, both the left and the right arm consist of first $r$ repetitions of character $c$,
then $i-1$ blocks consisting of $k$ repetitions of characters $c+1,c+2,\ldots,c+i-2$, and then finally $\ell$ repetitions of character $c+i-1$.
Thus, the left and the right arm are always order-isomorphic. Further, for every choice of $i$ and the starting position we obtain a different word,
as two such fragments of the same length either start with different characters or differ in the length of the first block of the same
character.

Now, we analyze the number of such {\opsquare}s in $w$. By considering every $1\leq i\leq \lfloor \sigma/2\rfloor$ and starting
position $1,2,\ldots,|s|-2ik+1$, we obtain that the number of {\opsquare}s in $w$ is at least:

\begin{align*}
\sum_{i=1}^{\lfloor \sigma/2\rfloor} (|s|-2ik+1) &= \sum_{i=1}^{\lfloor\sigma/2\rfloor} (\sigma k-2ik+1)  = \lfloor\sigma/2\rfloor \cdot (\sigma k - k (\lfloor \sigma/2\rfloor + 1) +1) \\
& \geq  \lfloor \sigma/2 \rfloor \cdot (k(\lceil \sigma/2\rceil -1) + 1) .
\end{align*}
For $\sigma \geq 3$, this is at least $\sigma^{2}k/12=\sigma n/12$ for any $k$.
For $\sigma=1,2$, we additionally assume $k\geq 2$ and count {\opsquare}s of the form $1^{2i}$,
there are $\lfloor k/2\rfloor \geq n/6 \geq \sigma n / 12$ of them.
Thus, in either case for every $k\geq 2$ we obtain a string of length $n=k\sigma$ over $\Sigma$
containing $\sigma n/12$ {\opsquare}s that are distinct as words.

\begin{theorem}
\label{thm:oplower}
For any alphabet $\Sigma=\{1,2,\ldots,\sigma\}$, there exists an infinite family of strings of length $n=k\sigma$ over $\Sigma$
containing $\Omega(\sigma n)$ {\opsquare}s distinct as words.
\end{theorem}


\section{Upper bound}
\label{upperop}
Our goal in this section is to upper bound the number of {\opsquare}s in a given string $w$ of length $n$ over
the alphabet $\Sigma=\{1,\ldots,\sigma\}$.
Recall that $uv$ is an {\opsquare} when $u\neq v$ and $u\op v$.
We will show that this number is $\Oh(\sigma n)$.
As explained in the introduction, by the result of Brlek and Li~\cite{Brlek}, the number of regular squares, that is,
fragments of the form $uu$ that are distinct as words, is less than $n$.
Thus, our result in fact allows us to upper bound the number of fragments of the form $uv$, where $u\op v$,
that are distinct as words by $\Oh(\sigma n)$ as well.

We consider each suffix of $w$ separately. For each suffix $w[i..n]$, we will upper bound the number of prefixes
of $w[i..n]$ that are {\opsquare}s by $\Oh(\sigma)$. Therefore, to avoid cumbersome notation in the remaining part of this section we will
assume that we have a string $s$ of length $m$ over the alphabet $\Sigma=\{1,\ldots,\sigma\}$, and we want
to upper bound the number of {\opsquare}s $uv$ that are prefixes of $s$ by $\Oh(\sigma)$.
See Figure~\ref{fig:squares}.

\begin{figure}[ht]
\centerline{\includegraphics[scale=.3]{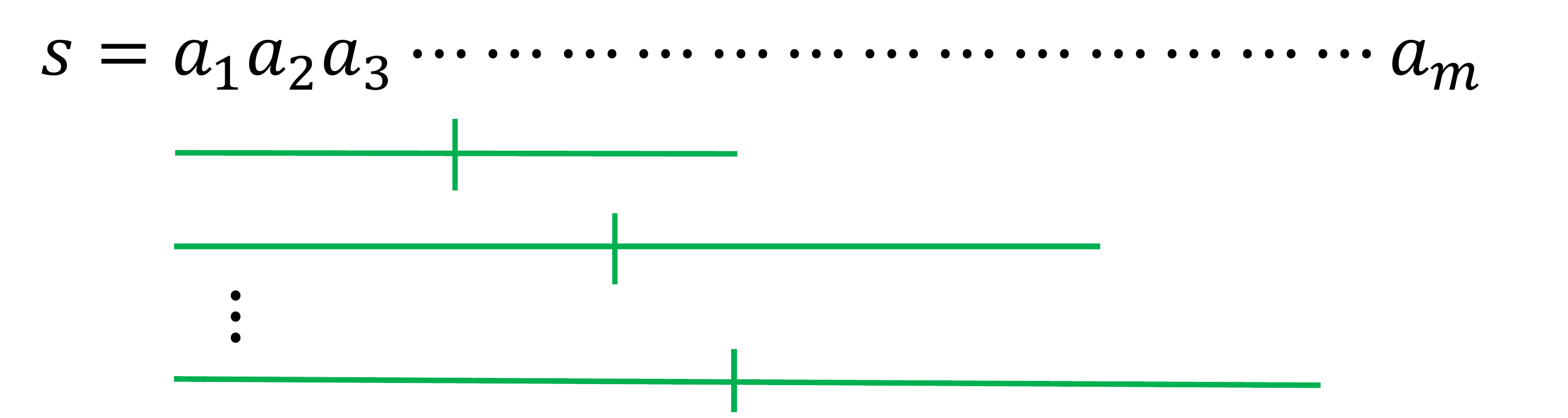}}
\caption{Green prefixes of $s$ are {\opsquare}s.}
\label{fig:squares}
\end{figure}

Kociumaka et al.~\cite{Kociumaka} observed that every {\opsquare} $uv$ that is a prefix of $s$ can be obtained
as follows (recall that in our definition $u\neq v$). We call $i$ a leftmost occurrence when $s[j]\neq s[i]$ for every $i<j$.
Then, there exists $i$ and $j$ such that both $i$ and $j$ are leftmost occurrences, $i$ belongs to $u$ and $j$
belongs to $v$, and further $|u|=j-i$. More formally:

\begin{proposition}[{\cite[Lemma 4.2 and Corollary 4.3]{Kociumaka}}]
\label{kociumaka}
We can construct an injective function $g$ mapping {\opsquare}s that are prefixes of $s$ to 2-element
subsets of the alphabet as follows. We choose the smallest $i$ belonging to $v$
such that $s[i]=a$ does not occur in $u$, and let $s[i-|u|]=b$ be its counterpart in $u$,
then set $g(uv)=\{a,b\}$.
Both $i$ and $i-|u|$ are leftmost occurrences.
\end{proposition}

We split all {\opsquare}s that are prefixes of $s$ into groups. Let $O_k$ denote the group of {\opsquare}s that
are prefixes of $s$ having length at least $2^k$ and at most $2^{k+1}-1$:
\begin{definition}
$ O_k=\{uv\,|\, u\neq v \text{ and } u\op v \text{ and } 2^k \leq |uv| < 2^{k+1}\} \text{ for } 0\leq k\leq\log m. $ 
\end{definition}
\noindent In other words, we split $s$ into consecutive ranges of exponentially increasing lengths,
such that the $k$-th range is of length $2^{k-1}$, starts at position $2^k$ and ends at position $2^{k+1}-1$ in $s$
(where $0\leq k\leq\log m$ and the final range may not be complete when $m < 2^{k+1}-1$).
Then, $O_{k}$ consists of {\opsquare}s that end in the $k$-th range. See Figure~\ref{fig:groups}.
\begin{figure}[ht]
\centerline{\includegraphics[scale=.3]{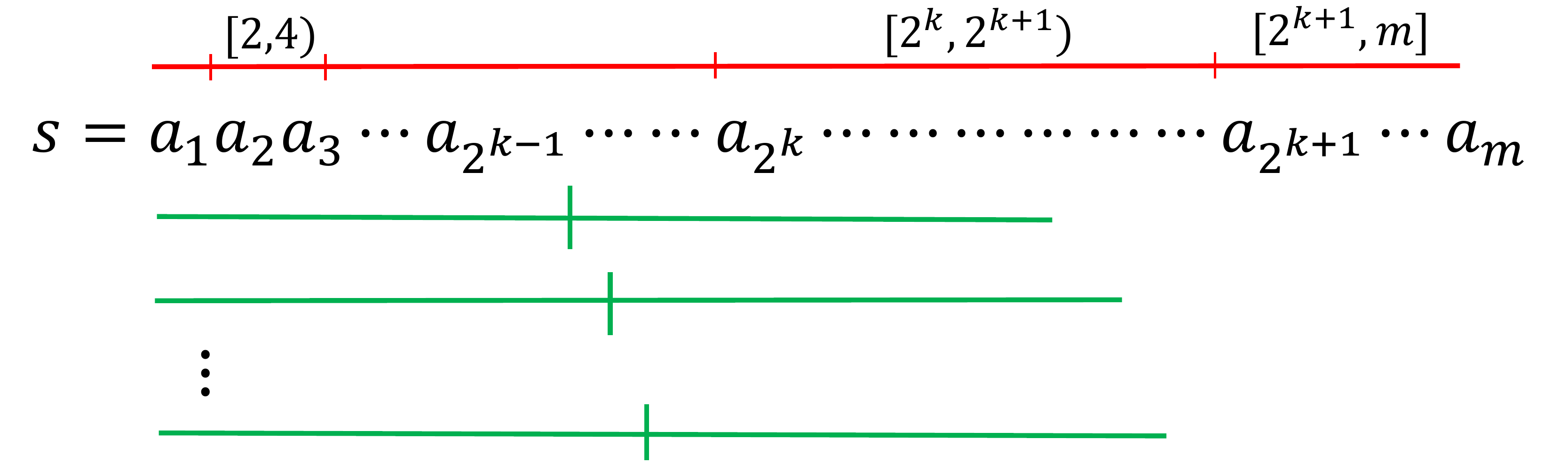}}
\caption{Green prefixes of $s$ are {\opsquare}s ending in the $k$-th range. The red line illustrates the ranges.}
\label{fig:groups}
\end{figure}

The number of {\opsquare}s $uv$ that are prefixes of $s$ is $\sum_{k=0}^{\log m} |O_k|$.
In order to upper bound the sum, we will separately upper bound the size of each group.
We first need some propositions.

\begin{proposition}
\label{prefixsuffix}
For any $\{uv,u'v'\}\in O_k$ such that $|u|<|u'|$ and $\Delta=|u'|-|u|$, $|u|-\Delta$ is an op-border of
both $u$ and $v$.
\end{proposition}
\begin{proof}
Because $u\op v$ it is enough to show that $\Delta$ is an op-border of $u$.
By the assumption that both $uv$ and $u'v'$ are {\opsquare}s we have:
\[ u[1..|u|-\Delta]=u'[1..|u|-\Delta]\op v'[1..|u|-\Delta]=v[\Delta+1..|u|]\op u[\Delta+1..|u|]  .\]
See Figure~\ref{fig:border}.
\end{proof}

\begin{figure}[ht]
\centerline{\includegraphics[scale=.3]{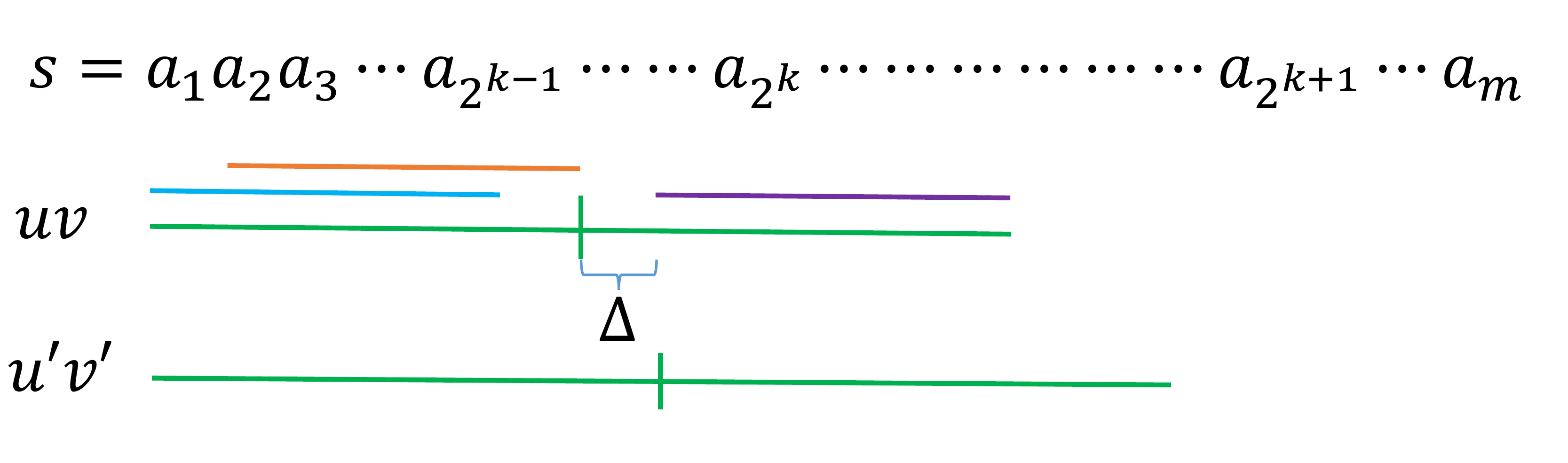}}
\caption{The green lines correspond to $uv$ and $u'v'$. The blue line corresponds to $u[1..|u|-\Delta]$. The orange line corresponds to $u[\Delta+1..|u|]$. The purple line corresponds to $v[\Delta+1..|u|]$.}
\label{fig:border}
\end{figure} 

In the remaining part of this section, we will often consider $\{uv,u'v'\}\in O_k$ such that $|u|<|u'|$ and $\Delta=|u'|-|u|$.
Then, by Proposition~\ref{prefixsuffix}, $\Delta$ is an op-border of $u$, and thus by Proposition~\ref{borderperiods}
$u$ can be represented as a concatenation of $f=\lfloor |u|/\Delta\rfloor$ blocks $b_{1},b_{2},\ldots,b_{f}$ and
one incomplete block $b_{f+1}$, where $|b_{1}|=|b_{2}|=\ldots=|b_{f}|=\Delta$ and $|b_{f+1}|<\Delta$,
such that $b_{1}b_{2}\ldots b_{f}[1..|b_{f+1}|]\op b_{2}b_{3}\ldots b_{f} b_{f+1}$, 
$b_{1}\op b_{2}\op\ldots\op b_{f}$ and
$b_{1}[1..|b_{f+1}|]\op b_{2}[1..|b_{f+1}|]\op\ldots\op b_{f}[1..|b_{f+1}|]\op b_{f+1}$. See Figure~\ref{fig:blocks}.
For brevity, in the remaining part of the paper we will describe this situation by saying that $\{uv,u'v'\}\in O_k$ define
blocks $b_{1},b_{2},\ldots,b_{f},b_{f+1}$.

\begin{figure}[ht]
\centerline{\includegraphics[scale=.3]{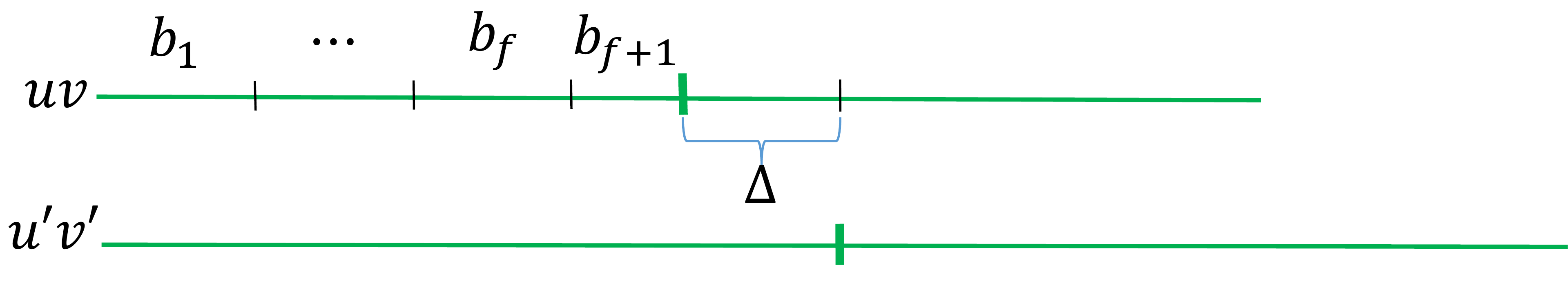}}
\caption{Blocks defined by $\{uv,u'v'\}$ in $u$.}
\label{fig:blocks}
\end{figure} 

\begin{proposition}
\label{delta}
If $|O_{k}|\geq 3$ then there exist $\{uv,u'v',u''v''\}\in O_k$ such that $0 < |u'|-|u|, |u''|-|u'| < 2^{k} / (|O_k|-2)$.
\end{proposition}
\begin{proof}
The length of every {\opsquare} in $O_{k}$ belongs to $[2^{k},2^{k+1})$, thus the length of its left arm falls within
$[2^{k-1},2^{k})$. Let $O_{k}={u_{1}v_{1},u_{2},v_{2},\ldots,u_{\ell}v_{\ell}}$ with $|u_{1}|<|u_{2}|<\ldots<|u_{\ell}|$.
Then, for some $i\in\{1,2,\ldots,\lfloor (\ell-1)/2 \rfloor\}$ we must have $|u_{2i+1}| < |u_{2i-1}| + 2^{k-1}/ \lfloor (\ell-1)/2 \rfloor$
(as otherwise we would have $u_{\ell} \geq u_{1}+2^{k-1}$).
The sought {\opsquare}s are $u_{2i-1}v_{2i-1}, u_{2i}v_{2i}, u_{2i+1}v_{2i+1}$ because:
\begin{align*}
|u_{2i}|-|u_{2i-1}|, |u_{2i+1}|-|u_{2i}| &< |u_{2i+1}|-|u_{2i-1}| < 2^{k-1}/ \lfloor (\ell-1)/2 \rfloor  \\
& \leq  2^{k-1} / (\ell/2 - 1) = 2^{k} / (|O_{k}|-2).
\end{align*}
This finishes the proof.
\end{proof}

With all the propositions in hand, we are now ready for the technical lemmas. Our goal is to upper bound $\sum_{k}|O_{k}|$
by the number of leftmost occurrences. To this end, we need to show that, if some $O_{k}$ is large then there are many
leftmost occurrences in some range. This will be done by applying the following reasoning to the three {\opsquare}s chosen
by applying Proposition~\ref{delta}.

\begin{lemma}
\label{leftposition}
If $|O_k|\geq 3$ then for any $\{uv,u'v',u''v''\}\in O_k$ where $|u|<|u'|<|u''|$ such that $\{uv,u'v'\}$ defines $b_1,\ldots,b_f,b_{f+1}$
and $\{u'v',u''v''\}$ defines $b'_1,\ldots,b'_f,b'_{f+1}$ there is a leftmost occurrence in $b_j$ such that $j\neq 1$ or
there is a leftmost occurrence in $b'_{j'}$ such that $j'\neq 1$.
\end{lemma}

\begin{proof}
Let $\Delta=|u'|-|u|$ be the length of every block $b_{j}$ and $\Delta'=|u''|-|u'|$ be the length of every block $b'_{j}$.
By Proposition~\ref{kociumaka}, we know that there must be a leftmost occurrence $i$ that falls within $u'$
and its corresponding leftmost occurrence $i+|u'|$ that falls within $v'$.
If $s[i]$ belongs to a block $b'_{j}$ with $j\neq 1$ then we are done. Thus, we assume that $i$ belongs
to $b'_{1}$. 
We claim that the leftmost occurrence $i+|u'|$ falls within $v$. To verify this, we calculate:
\[
i+|u'| \leq \Delta' + |u'| = |u''| - |u'| + |u'| = |u''| < 2^{k} \leq |uv| .
\]
We have established that $i+|u'|$ is a leftmost occurrence and falls within $v$. Thus,
$s[i'] \neq s[i+|u'|]$ for every $i' \in [1, i+|u'|)$. Because $u\op v$, this then implies that 
$s[i'-|u|] \neq s[i+|u'|-|u|]$ for every $i' \in [|u|+1,i+|u'|)$.  Thus,
$i+|u'| - |u|$ is also a leftmost occurrence. We claim that $s[i+|u'|-|u|]$ cannot belong
to $b_{1}$. To verify this, we calculate:
\[
i+|u'|-|u| \geq 1 + |u'| - |u| = 1 + \Delta .
\]
Thus, we have found a leftmost occurrence $i+|u'|-|u|$ that falls within $u$ and belongs to a block $b_{j}$ with $j\neq 1$.
See Figure~\ref{fig:leftBj}.
 \end{proof}

\begin{figure}[ht]
\centerline{\includegraphics[scale=.9]{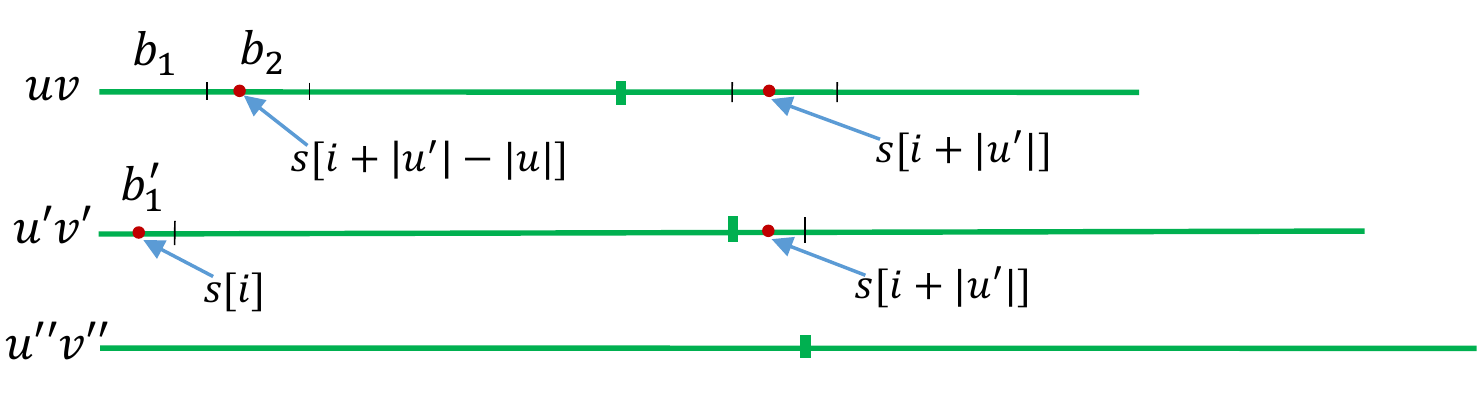}}
\caption{The red points correspond to the leftmost occurrences considered in the proof of Lemma~\ref{leftposition}.}
\label{fig:leftBj}
\end{figure} 

Next, we show that if $\{uv,u'v'\}\in O_{k}$ define blocks $b_{1},b_{2},\ldots,b_{f},b_{f+1}$ such that
there is a leftmost occurrence in block $b_{j}$ for some $j\neq 1$ then, in fact, there is a leftmost occurrence in every
block $b_{j}$. This reasoning is done in two steps.

\begin{lemma}
\label{leftblocks}
Let $b$ be an op-border of $u=s[1..|u|]$ that defines blocks $b_{1},b_{2},\ldots,b_{f},b_{f+1}$, and assume that there is a leftmost occurrence
in block $b_{j}$, for some $j\in [1,f+1]$. Then there is a leftmost occurrences in every block $b_{1},b_{2},\ldots,b_{j}$.
\end{lemma}

\begin{proof}
Let $\Delta=|b_{1}|=|b_{2}|=\ldots =|b_{f}|$ and $|b_{f+1}|<\Delta$. By induction, it is enough to show that if there is
a leftmost occurrence in $b_{j}$ for some $j\geq 2$ then there is a leftmost occurrence in $b_{j-1}$. Let $i$ be a leftmost
occurrence that belongs to $b_{j}$. Then $u[i']\neq u[i]$ for every $i'\in [1,i)$. Because
$b_{1}b_{2}\ldots b_{f-1}b_{f}[1..|b_{f+1}|] \op b_{2}b_{3}\ldots b_{f}b_{f+1}$, this implies $u[i'-\Delta] \neq u[i-\Delta]$ for
every $i'\in [\Delta+1,i)$. But then $i-\Delta$ is also a leftmost occurrence, and it belongs to $b_{j-1}$ as required.
See Figure~\ref{fig:leftblocks}.
 \end{proof}

\begin{figure}[ht]
\centerline{\includegraphics[scale=.3]{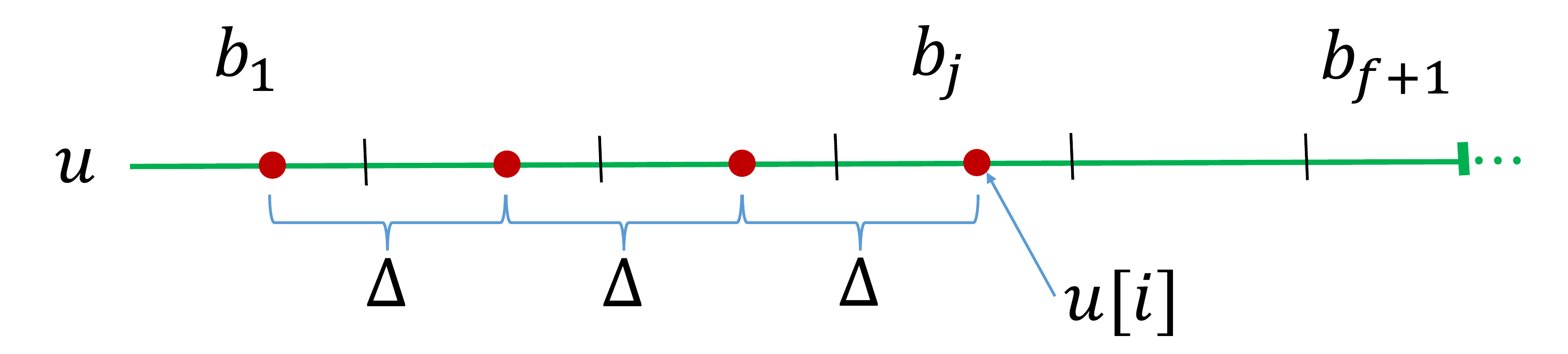}}
\caption{Each red point corresponds to a leftmost character at the same relative position in every block $b_{1},b_{2},\ldots,b_{j}$.}
\label{fig:leftblocks}
\end{figure}

\begin{lemma}
\label{rightblocks}
Let $b$ be an op-border of $u=s[1..|u|]$ that defines blocks $b_{1},b_{2},\ldots,b_{f},b_{f+1}$, and assume that there is a leftmost character
in block $b_{2}$. Then there is a leftmost character in every block $b_{1},b_{2},\ldots,b_{f}$.
\end{lemma}

\begin{proof}
Let $\Delta=|b_{1}|=|b_{2}|=\ldots |b_{f}|$ and $|b_{f+1}|<\Delta$.
By assumption, there is a leftmost character $s[i]$ in block $b_{2}$, that is, $i\in [\Delta+1,2\Delta]$.
Our goal is to show that there is a leftmost character in every block $b_{1},b_{2},\ldots,b_{f}$.

Because $b_{1}b_{2}\ldots b_{f}[1..|b_{f+1}|]\op b_{2}b_{3}\ldots b_{f} b_{f+1}$,
each position $x\in [1..\Delta]$ satisfies exactly one of the following possibilities:
\begin{enumerate}
\item $s[x+p\cdot\Delta]$ is the same, for all integers $p\in [0,f)$,
\item $s[x+p\cdot\Delta]<s[x+(p+1)\cdot\Delta]$ for all integers $p\in [0,f-1)$,
\item $s[x+p\cdot\Delta]>s[x+(p+1)\cdot\Delta]$ for all integers $p\in [0,f-1)$.
\end{enumerate}
Note that $i_{0}=i-\Delta$ satisfies (2) or (3), because $s[i]$ is different than $s[1],s[2],\ldots,s[i-1]$, so in particular $s[i-\Delta]\neq s[i]$.
By reversing the order of the alphabet, it is enough to establish the lemma assuming that $i_{0}$ satisfies (2)
Furthermore, we choose some positions $i_{1},i_{2},\ldots,i_{\ell}$ in $b_1$ as follows.
Let $C$ be the set of characters that appear in $b_{1}$.
$i_{1}\in [1,\Delta]$ is chosen so that $s[i_{1}]$ is the strict successor of $s[i_{0}]$ in $C$,
then $i_{2}\in [1,\Delta]$ is chosen so that $s[i_{2}]$ is the strict successor of $s[i_{1}]$ in $C$,
and so on. If there are multiple choices for the next $i_{j}\in [1,\Delta]$ then we take the smallest. We stop when one of the following two
possibilities holds:
\begin{enumerate}[(a)]
\item $s[i_{\ell+1}]$ is not defined, i.e. $s[i_{\ell}]$ is the largest character in $b_{1}$.
\item $i_{\ell+1}$ satisfies (1) or (3).
\end{enumerate}
Note that, by definition, $i_0,i_1,\ldots,i_\ell$ all satisfy $(2)$. Further, $s[i_{0}]$ is a leftmost character
because $s[i]$ is a leftmost character, so $s[i']\neq s[i]$ for every $i'\in [1,i)$, and $b_{1}\op b_{2}$
so $s[i'-\Delta]\neq s[i-\Delta]$ for every $i'\in [\Delta+1,i)$.
Next, $s[i_{1}],\ldots,s[i_{\ell}]$ are all leftmost characters
because we are always choosing the smallest $i_{j}$ such that $s[i_{j}]$ is equal to a specific character, for $j=1,2,\ldots,\ell$.

We summarize the situation so far. For every integer $p\in [0,f)$, $s[i_{\ell}+p\cdot \Delta]$ belongs to block $b_{p+1}$,
and we want to show that it is a leftmost character. We know that $s[i_{\ell}]$ is a leftmost character, thus by $b_{1}\op b_{p+1}$
we obtain that $s[i_{\ell}+p\cdot \Delta]$ does not occur earlier in $b_{p+1}$. We need to establish that it also does not
occur earlier in $b_{1},b_{2},\ldots,b_{p}$. We separately consider the two possible cases (a) and (b).

\begin{enumerate}[(a)]

\item $s[i_{\ell+1}]$ is not defined, i.e. $s[i_{\ell}]$ is the largest character in $b_{1}$. We know that
$i_{\ell}$ satisfies (2), so $s[i_\ell]<s[i_{\ell}+\Delta]<\ldots<s[i_{\ell}+(p-1)\cdot \Delta] <s[i_{\ell}+p\cdot\Delta]$.
For all integers $q\in [0,p)$, by $b_{1}\op b_{q+1}$ we obtain that $s[i_{\ell}+q\cdot\Delta]$ is the largest character in $b_{q+1}$.
So in fact $s[i_{\ell}+p\cdot \Delta]$ is larger than all characters in the whole block $b_{q+1}$, for every integer $q\in [0,p)$,
making $i_{\ell}+p\cdot \Delta$ a leftmost occurrence.

\item $i_{\ell+1}$ is defined and satisfies (1) or (3),
so $s[i_{\ell+1}] \geq s[i_{\ell+1}+\Delta] \geq \ldots\geq s[i_{\ell+1}+(p-1)\cdot \Delta] \geq s[i_{\ell+1}+p\cdot\Delta]$. See Figure~\ref{fig:lemma12b}.
We know that $i_{\ell}$ satisfies (2), so $s[i_\ell]<s[i_{\ell}+\Delta]<\ldots<s[i_{\ell}+(p-1)\cdot \Delta] <s[i_{\ell}+p\cdot\Delta]$.
Recall that $s[i_{\ell+1}]$ is a strict successor of $s[i_{\ell}]$ in $b_{1}$. Thus, for every $i'\in [1,\Delta]$ we have
that $s[i']$ does not belong to the interval $(s[i_{\ell}],s[i_{\ell+1}])$. Because we have $b_{1}\op b_{q}$, for every integer $q\in [0,p)$, this implies
$s[i'+q\cdot \Delta]$ does not belong to the interval $(s[i_{\ell}+q\cdot \Delta],s[i_{\ell+1}+q\cdot \Delta])$.
As observed earlier, $s[i_{\ell}+q\cdot \Delta] < s[i_{\ell}+p\cdot \Delta]$ and $s[i_{\ell+1}+q\cdot \Delta] \geq s[i_{\ell+1}+p\cdot \Delta]$.
We conclude that, for every $i'\in [1,\Delta]$, we have that $s[i'+q\cdot \Delta]$ does not belong to the interval
$[s[i_{\ell}+p\cdot \Delta],s[i_{\ell+1}+p\cdot \Delta])$
(the interval is non-empty, as both positions belong to the same block $b_{q}$, and by $b_{q}\op b_{1}$ we have that
$s[i_{\ell+1}+q\cdot \Delta]$ is a strict successor of $s[i_{\ell}+q\cdot \Delta]$ in $b_{q}$).
In particular, $s[i'+q\cdot \Delta]\neq s[i_{\ell}+p\cdot \Delta]$, so $s[i_{\ell}+p\cdot \Delta]$ does not occur in $b_{q+1}$,
making it a leftmost character.
\end{enumerate}

We have established that, for every integer $p\in [0,f)$, $i_{\ell}+p\cdot \Delta$ is a leftmost occurrence.
\end{proof}
\begin{figure}[ht]
\centerline{\includegraphics[scale=.3]{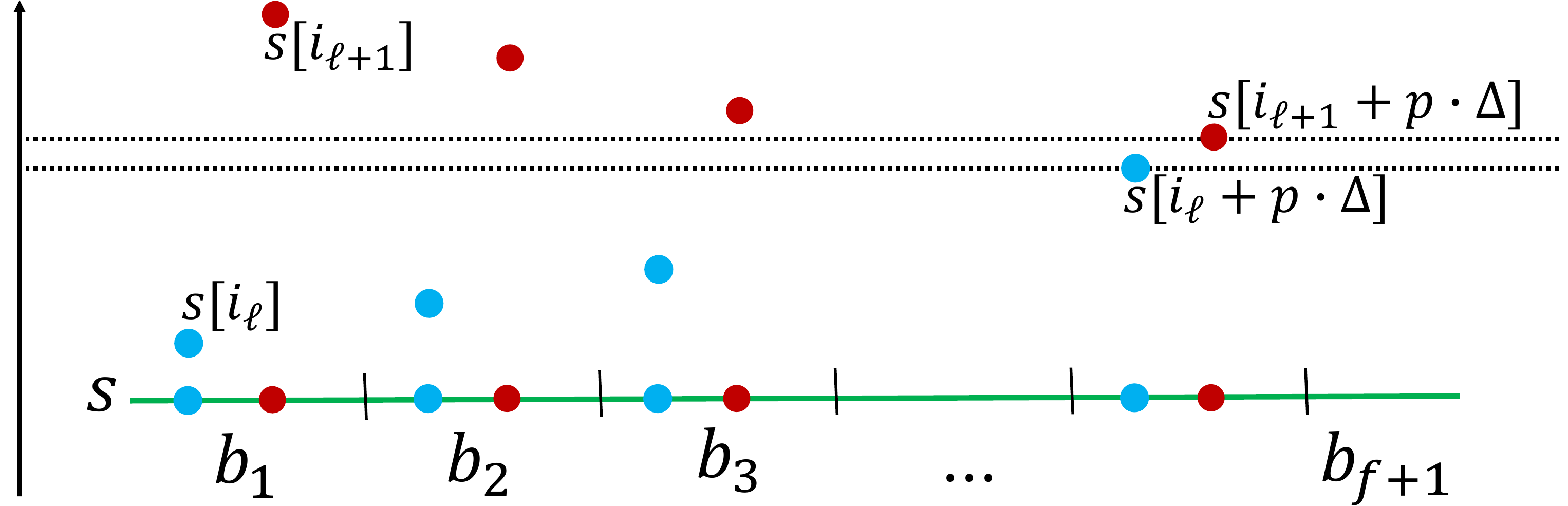}}
\caption{The red points correspond to $s[i_{\ell+1}],...,s[i_{\ell+1}+p\cdot\Delta]$. The blue points correspond to $s[i_{\ell}],...,s[i_{\ell}+p\cdot\Delta]$. The black arrow illustrates the character's axis.}
\label{fig:lemma12b}
\end{figure}
By combining the above lemmas we obtain the following conclusion.
\begin{lemma}
\label{allblocks}
If $|O_k|\geq 3$ then for any $\{uv,u'v',u''v''\}\in O_k$ where $|u|<|u'|<|u''|$ such that $\{uv,u'v'\}$ defines $b_1,\ldots,b_f,b_{f+1}$
and $\{u'v',u''v''\}$ defines $b'_1,\ldots,b'_f,b'_{f+1}$ there is a leftmost occurrence in every block $b_1,b_{2},\ldots,b_{f}$
or there is a leftmost occurrence in every block $b'_{1},b'_{2},\ldots,b'_{f'}$.
\end{lemma}
\begin{proof}
Recall that by Lemma~\ref{prefixsuffix}, $|u'|-|u|$ is an op-border of $u=s[1..|u|]$ while $|u''|-|u'|$ is an op-border of $u'=s[1..|u'|]$.
By Lemma~\ref{leftblocks} there is a leftmost occurrence in $b_2$ or in $b'_2$. Then, by Lemma~\ref{rightblocks} applied either
to the blocks $b_{1},b_{2},\ldots,b_{f},b_{f+1}$ defined by the op-border $|u'|-|u|$ or the blocks $b'_{1},b'_{2},\ldots,b'_{f'},b'_{f'+1}$
defined by the op-border $|u''|-|u'|$, there is a leftmost occurrence in every block $b_{1},b_{2},\ldots,b_{f}$ or in every block
$b'_{1},b'_{2},\ldots,b'_{f'}$.
\end{proof}

We are now ready to upper bound $\sum_{k}|O_{k}|$ by the number of leftmost characters. We will show that, if some $O_{k}$
is large then there are many leftmost characters in some range. To this end, we define groups of leftmost occurrences.
Let $L_{k}$ consist of the leftmost occurrences $i$ such that $i\in [2^{k},2^{k+1})$:
\begin{definition}
$ L_k=\{ i \,|\, i\in[2^k,2^{k+1}) \wedge \forall_{j\in[1,i)} s[i]\neq s[j] \} \text{ for } 0\leq k\leq\log m. $
\end{definition}
Note that the groups are disjoint, i.e. $L_{k}\cap L_{k'}=\emptyset$ for any $k$ and $k'$. Thus $\sum_{k=0}^{\log m} |L_k| \leq \sigma$.
With this definition in hand, we are ready to show the main technical lemma.

\begin{lemma}
\label{op_in_s}
The number of {\opsquare}s that are prefixes of $s$ is $\Oh(\sigma)$.
\end{lemma}
\begin{proof}
To establish the lemma we want to connect $|O_{k}|$ with $|L_{k}|$, and then sum over all possible values of $k$.
$k=0,1$ will be considered separately, and for larger $k$ we apply different arguments for $|O_{k}|\geq 11$ and $|O_{k}| \leq 10$.

We first consider $k\geq 2$ such that $|O_{k}| \geq 11$. In particular, $|O_{k}| \geq 3$, so
by Proposition~\ref{delta}, there  exist $\{uv,u'v',u''v''\}\in O_k$ such that
$0 < |u'|-|u|, |u''|-|u'| < 2^{k} / (|O_k|-2)$.
By Lemma~\ref{allblocks}, for any $\{uv,u'v',u''v''\}\in O_k$ where $|u|<|u'|<|u''|$ such that $\{uv,u'v'\}$ defines $b_1,\ldots,b_f,b_{f+1}$
in $u$ and $\{u'v',u''v''\}$ defines $b'_1,\ldots,b'_{f'},b'_{f'+1}$ in $u'$ either there is a leftmost occurrence in every block
$b_{1},b_{2},\ldots,b_{f}$ or there is a leftmost occurrence in every block $b'_{1},b'_{2},\ldots,b'_{f'}$
In either case, we have found $\{uv,u'v'\}\in O_{k}$ with $0 < \Delta < 2^{k} / (|O_k|-2)$,  where $\Delta=|u'|-|u|$, such that
$\{uv,u'v'\}$ defines $b_1,\ldots,b_f,b_{f+1}$ with $f=\lfloor |u| / \Delta \rfloor $
and there is a leftmost occurrence in every block $b_{1},b_{2},\ldots,b_{f}$.
We want to establish a lower bound on the number of leftmost occurrences in $L_{k-2}$.
To this end, it is enough to show a lower bound on the number of blocks $b_{i}$ that are fully contained
in the range $[2^{k-2},2^{k-1})$. Recall that $|u|\in [2^{k-1},2^{k})$, and $u=b_{1}b_{2}\ldots b_{f}b_{f+1}$.
Thus, $u[2^{k-2}..2^{k-1}-1]$ consists of a suffix (possibly empty) of some $b_{j}$, then $b_{j+1}, b_{j+2},\ldots,b_{j+\ell}$,
and then a prefix of $b_{j+\ell+1}$ (where $b_{j+\ell+1}$ might be the incomplete block $b_{f+1}$ that should not be counted
in the lower bound). Thus, the number of blocks $b_{i}$ that are fully contained in the range $[2^{k-2},2^{k-1})$
is at least $\lfloor 2^{k-2}/\Delta \rfloor-1$. See Figure~\ref{fig:LK-2}. Combining this with the upper bound
on $\Delta$, we obtain the following inequality:
\[ |L_{k-2}| \geq \left\lfloor\dfrac{2^{k-2}}{\Delta}\right\rfloor - 1 \geq \dfrac{2^{k-2}}{\Delta}-2 > \dfrac{|O_k|-2}{4}-2 = \dfrac{|O_{k}|-10}{4} . \]
Using the assumption $|O_{k}| \geq 11$, we conclude that $|L_{k-2}| > |O_{k}|/44$. Hence:
\[ \sum_{k \geq 2 : |O_{k}| \geq 11} |O_{k}| < \sum_{k \geq 2} 44 \cdot |L_{k-2} | \leq 44\sum_{k} |L_{k} | \leq 44 \cdot \sigma .
\]

\begin{figure}[ht]
\centerline{\includegraphics[scale=.43]{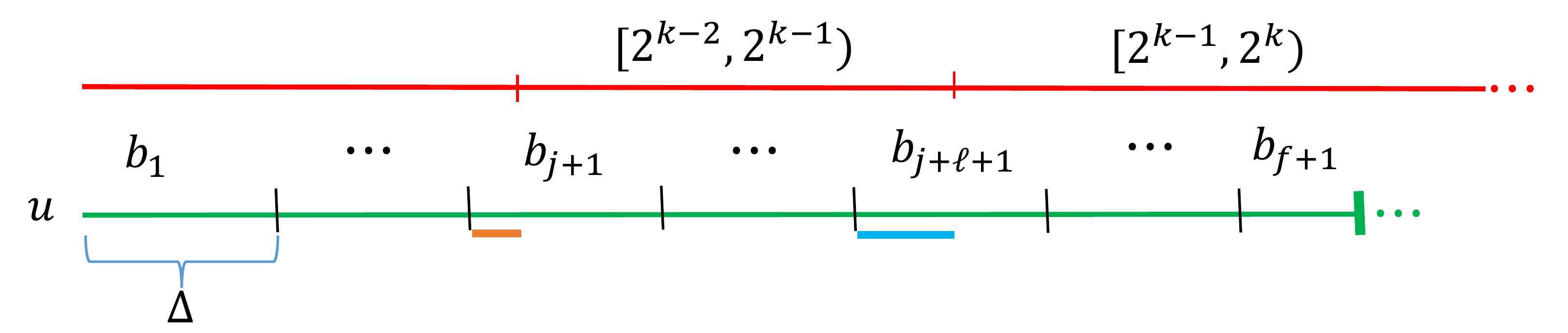}}
\caption{The orange line corresponds to the suffix of $b_{j}$. The blue line corresponds to the prefix of $b_{j+\ell+1}$. The red
line illustrates the ranges.}
\label{fig:LK-2}
\end{figure}

Next, we consider $k\geq 2$ such that $|O_{k}| \leq 10$. Of course, we have the trivial upper bound
$\sum_{k : |O_{k}| \leq 10} |O_{k}| \leq \sum_{k=0}^{\log m}10=\Oh(\log m)$.
As in the previous case, we want to use the leftmost occurrences to improve the bound.
Recall that, by Proposition~\ref{kociumaka}, every {\opsquare} $uv\in O_{k}$ is defined by a pair of leftmost occurrences $i$ and $j$,
where $i$ belongs to $u$ and $j$ belongs to $v$. Because $|uv|\in [2^{k},2^{k+1})$, we conclude that $j$
falls within the range $[2^{k-1}+1,2^{k+1})$, so must belong to $L_{k-1}\cup L_{k}$.
Hence, $O_{k}$ can be non-empty only when $L_{k-1}$ or $L_{k}$ is non-empty. Hence:
\begin{alignat*}{3}
\sum_{k \geq 2 : |O_{k}| \leq 10} |O_{k}| \leq \sum_{k \geq 2 : |O_{k}|>0} 10 &\leq \sum_{k \geq 2 : |L_{k-1}|>0} 10 &&+ \sum_{k \geq 2 : |L_{k}|>0} 10 && \\
& \leq  10 \sum_{k\geq 2} |L_{k-1}| &&+ 10\sum_{k\geq 2} |L_{k}| &&\leq 20\sigma .
\end{alignat*}

To upper bound $\sum_{k} |O_{k}|$, we split the sum into three parts. For $k=0,1$, we have $|O_{0}| \leq 1$ and $|O_{1}| \leq 2$.
Then, for $k\geq 2$ we separately consider all $k$ with $|O_{k}| \geq 11$ and $|O_{k}| \leq 10$ and plug in the above upper bounds.
Overall, we obtain:
\[ \sum_{k} |O_{k}| \leq 1 + 2 + 44\cdot \sigma + 20\cdot \sigma = \Oh(\sigma). \]
Thus, the number of {\opsquare}s that are prefixes of $s$ is $\Oh(\sigma)$.  
\end{proof}

We conclude the section with the main theorem.
\begin{theorem}
\label{optotal}
The number of {\opsquare}s in a string $s$ of length $n$ over an alphabet of size $\sigma$ is $\Oh(\sigma n)$.
\end{theorem}
\begin{proof}
We consider each suffix of $s$ separately. For each suffix $s[i..n]$, we apply Lemma~\ref{op_in_s} to conclude that
the number of {\opsquare}s that are prefixes of $s[i..n]$ is upper bounded by $\Oh(\sigma)$. Thus, summing over all
$i$ we obtain that the number of {\opsquare}s in $s$ is $\Oh(\sigma n)$.
\end{proof}


\section{Algorithm}
\label{algo}
In this section, we describe the algorithm that reports all occurrences {\opsquare}s in a string $w[1..n]$
over an alphabet of size $\sigma$ in $\Oh(\sigma n)$ time.

The high-level idea of the algorithm is to generate $\Oh(\sigma n)$ candidates for {\opsquare}s and then test each of them
in constant time. 
To this end, we first describe a mechanism for checking if $w[i..i+\ell-1] \op  w[i+\ell..i+2\ell-1]$ in constant
time. This can be implemented with an LCA query on the order-preserving suffix tree of $w$, as explained in~\cite{OpSuffix}.
However, we need to explain how to construct this structure in $\Oh(\sigma n)$ time.

Following~\cite{OpSuffix}, for a string $w[1..n]$ we define $\code(w)$ as $(\phi(w,1),\phi(w,2),\ldots,\phi(w,n))$,
where $\phi(w,i)=(\prev_{<}(w,i),\prev_{=}(w,i))$ and $\prev_{<}(w,i)=|\{k < i : w[k]<w[i] \}|$, $\prev_{=}(w,i)=|\{k < i : w[k]=w[i] \}|$.
We observe that $\code(w)=\code(w')$ if and only if $w\op w'$.
Then, the order-preserving suffix tree of $w[1..n]$ is the compacted trie of all strings of the form $\code(w[i..n])\#$,
for $i=1,2,\ldots,n$. It is easy to see that $w[i..i+\ell-1] \op  w[i+\ell..i+2\ell-1]$ if and only if the lowest
common ancestor of the leaves corresponding to $\code(w[i..n])\$$ and $\code(w[i+\ell..n])\$$ is at string depth at least
$\ell$. Therefore, assuming that we have already built the order-preserving suffix tree of $w[1..n]$, such a test can be
implemented in constant time after $\Oh(n)$ preprocessing for LCA queries~\cite{HarelT84}. It remains to explain how to
construct the order-preserving suffix tree. We stress that while \cite{OpSuffix} does provides an efficient $\Oh(n\log n/\log\log n)$
time construction algorithm (in fact, the full version~\cite{CROCHEMORE2016122} further improves the time complexity
to $\Oh(n\sqrt{\log n})$), such complexity is incompatible with our goal.

\begin{restatable}{lemma}{opsuffixtree}
\label{op_suffix_tree}
Given a string $w[1..n]$ over an alphabet of size $\sigma$, we can construct its order-preserving suffix tree in $\Oh(\sigma n)$ time and space.
\end{restatable}

\begin{proof}
As explained in~\cite{OpSuffix}, the order-preserving suffix tree of $w[1..n]$ can be constructed using the general framework
of Cole and Hariharan~\cite{ColeH00} for constructing a suffix tree for a quasi-suffix collection of strings $w_{1},w_{2},\ldots,w_{n}$.
The running time of their algorithm is $\Oh(n)$ with almost inverse exponential failure probability, assuming that one can
access the $j$-th character of any $w_{i}$ in constant time. The mechanism for accessing the $j$-th character of $w_{i}$
is called the \textit{character oracle}. In this particular application, $w_{i}=\phi(w[i..n])\#$. We will first describe how to
implement a constant-time character oracle for such strings, and then
explain why randomization is not needed in our setting.

We need to implement a new \textit{character oracle} that returns $\phi(w[i..n],j)$, for any $i,j$, in constant time after
$\Oh(\sigma n)$ time and space preprocessing. This requires being able to calculate $\prev_{<}(w[i..n],j)$ and $\prev_{=}(w[i..n],j)$
in constant time. To this end, we define a two-dimensional array $\cnt[i,x]=|\{ k : k<i, w[k] < x \}|$,
for $i=0,1,\ldots,n$ and $x=0,1,\ldots,\sigma$. All entries
in this array can be computed in $\Oh(\sigma n)$ total time and space. Then, we can calculate any $\prev_{<}(w[i..n],j)$ and $\prev_{=}(w[i..n],j)$ as follows:
\begin{align*}
\prev_{<}(w[i..n],j) =& \cnt[i+j-2,w[i+j-1]] - \cnt[i-1,w[i+j-1]]\\
\prev_{=}(w[i..n],j) =& (\cnt[i+j-2,w[i+j-1]] - \cnt[i+j-2,w[i+j-1]-1]) \\
& - (\cnt[i-1,w[i+j-1]] - \cnt[i-1,w[i+j-1]-1]) .
\end{align*}

To remove randomization, we observe that its only source in the algorithm of Cole and Hariharan is the need to
maintain, for each explicit node of the current tree, a dictionary indexed by the next character on an outgoing edge.
If we could show that there are at most $\Oh(\sigma)$ such edges, then the dictionary could be implemented as a simple list,
increasing the construction time to $\Oh(\sigma n)$, which is within our claimed bound.

Consider a non-leaf node $v$ of the current tree. It corresponds to a proper prefix of some $\code(w[i..n])\#$, which by
the definition of $\code(.)$ is equal to $\code(w[i..j])$, for some $j$. Let $c_{1}<c_{2}<\ldots <c_{k}$ be the distinct characters
of $w[i..j]$, and denote by $occ_{x}$ the number of occurrences of $c_{x}$ in $w[i..j]$. Now consider an edge outgoing from
$v$, and let $\code(w[i'..j'+1])$ correspond to the first node (implicit or explicit) after $v$ there. We know that
$\code(w[i'..j'])=\code(w[i..j])$, so the distinct characters of $w[i'..j']$ are $c'_{1}<c'_{2}<\ldots<c'_{k}$ with
$occ_{x}$ being the number of occurrences of $c'_{x}$ in $w[i'..j']$. Then, we analyze the possible values of
$(\prev_{<}(w[i'..j'+1],j'-i'+2),\prev_{=}(w[i'..j'+1],j'-i'+2))$, that is, the first character on the considered edge.
The first number is always equal to $\sum_{y=1}^{x-1}c_{y}$, for some $x\in [1,k+1]$. Then, the second
number is either $0$ or $occ_{x}$. Thus, overall we have only $2k\leq 2\sigma$ possible first characters,
which bounds the degree of any $v$ by $\Oh(\sigma)$.
\end{proof}

The main part of the algorithm is efficiently generating $\Oh(\sigma n)$ candidates for {\opsquare}s. Then,
each of them is tested in constant time as explained above, assuming the preprocessing from Lemma~\ref{op_suffix_tree}.
As in the proof of the $\Oh(\sigma n)$ upper bound on the number of {\opsquare}s, we will consider the suffixes
of the input string $w[1..n]$ one-by-one. Let $s=w[i..n]$ be the currently considered suffix, and $x_{1}<x_{2}<\ldots<x_{t}$
be the leftmost occurrences in $s$. By spending $\Oh(\sigma)$ time per each suffix, we can assume
that the positions $x_{1},x_{2},\ldots,x_{t}$ are known, as after moving from $w[i..n]$ to $w[i-1..n]$ we only have
to insert the new leftmost occurrence $i-1$ and possibly remove the previous leftmost occurrence $i'$ such
that $w[i-1]=w[i']$ (unless $w[i-1]$ has not been seen before), which can be done in $\Oh(t)=\Oh(\sigma)$ time.
By Proposition~\ref{kociumaka}, every prefix of $s$ that is an {\opsquare} can be obtained by choosing two leftmost
characters at positions $x_{i}$ and $x_{j}$, where $i<j$, and setting the length of the possible square to be $2(x_{j}-x_{i})$. This
gives us $\Oh(\sigma^{2})$ candidates for prefixes that could be {\opsquare}s. However, our goal is to
generate only $\Oh(\sigma)$ such candidates. Recall that all leftmost occurrences are partitioned into
groups $L_{0},L_{1},\ldots$. We first prove that it is enough to consider $x_{j}$ that is the smallest or the largest
element in its group.

\begin{figure}[ht]
\centerline{\includegraphics[scale=.43]{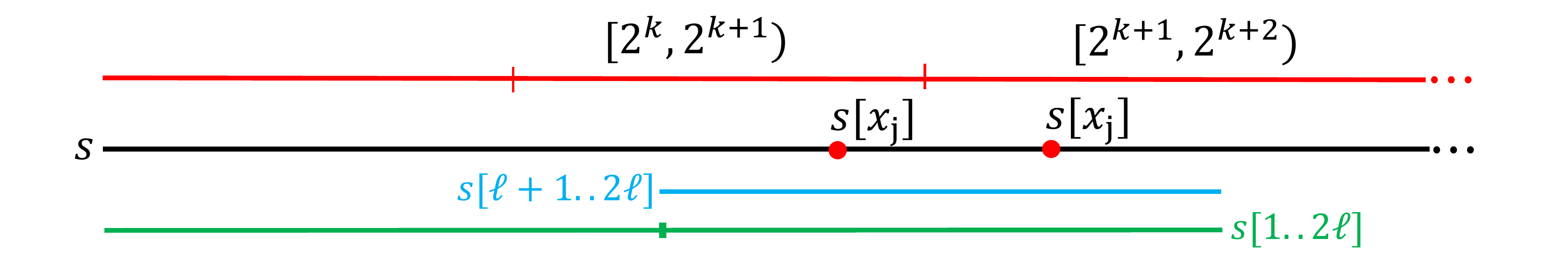}}
\caption{The {\opsquare} $s[1..2\ell]$ is colored in green. The red points are the two possibilities for $s[x_j]$. $s[\ell+1..2\ell]$ is colored in blue. The red line illustrates the ranges.}
\label{fig:lemma17}
\end{figure}

\begin{lemma}
\label{smallerorlargest}
Consider an {\opsquare} $s[1..2\ell]$. Then there exists $i<j$ such that $x_{j}\in [\ell+1,2\ell]$, $x_{j}-x_{i}=\ell$
and $x_{j}$ is either the smallest or the largest element of its group.
\end{lemma}

\begin{proof}
By Proposition~\ref{kociumaka}, we know that there is a leftmost occurrence in $s[\ell+1..2\ell]$. Choose the largest $k$
such that $2^{k}\leq \ell$ (so $2^{k+1} > \ell$). Consider two ranges $[2^{k},2^{k+1})$ and $[2^{k+1},2^{k+2})$
corresponding to groups $L_{k}$ and $L_{k+1}$, respectively. Because $2^{k} \leq \ell$ and $2^{k+1} > \ell$, we have
$2^{k} < \ell+1$, $2^{k+1} \in [\ell+1,2\ell]$ and $2\ell < 2^{k+2}$. Consequently, $s[\ell+1..2\ell]$ can be represented
as the concatenation of a suffix of $s[2^{k}..2^{k+1})$ and a prefix of $s[2^{k+1}..2^{k+2})$. The leftmost occurrence
that falls within $s[\ell+1..2\ell]$ belongs to the suffix or the prefix. See Figure~\ref{fig:lemma17}. If it falls within the suffix, the largest element of $L_{k}$ belongs to $[\ell+1,2\ell]$. If it falls within the prefix, the smallest element of $L_{k+1}$ belongs to $[\ell+1,2\ell]$.
Let $x_{j}\in [\ell+1,2\ell]$ be the corresponding leftmost occurrence. To complete the proof we need to establish that there
exists $i<j$ such that $x_{j}-x_{i}=\ell$. $s[x_{j}]$ is distinct from all $s[1],s[2],\ldots,s[x_{j}-1]$, and by $s[1..\ell]~\op s[\ell+1..2\ell]$ we obtain that $s[x_{j}-\ell]$ is distinct from all $s[1],s[2],\ldots,s[x_{j}-\ell-1]$. Thus, $x_{j}-\ell$ is a leftmost occurrence, hence $x_{j}-\ell=x_{i}$ for some $i<j$ as required.
\end{proof}

To generate the candidates, we iterate over all $j$ such that $x_{j}$ is the smallest or largest element of its
group $L_{k}$. Consider $i<j$ such that $x_{j}-x_{i}=\ell$ and $x_{j}\in [\ell+1,2\ell]$ for an {\opsquare}
$s[1..2\ell]$. Then, because $x_{j} \in [2^{k},2^{k+1})$, $\ell \geq 2^{k-1}$. To avoid clutter, let $s'=s[1..2^{k-1}]$.
Because $s[1..2\ell]$ is assumed to be an {\opsquare}, $s[\ell..\ell+2^{k-1}-1] \op  s'$.
This suggests the following natural strategy to generate the candidates: we iterate over all fragments $s[y..y+2^{k-1}-1]$
such that $x_{j}\in [y,y+2^{k-1}-1]$ and $s'\op  s[y..y+2^{k-1}-1]$, and output $s[1..2(y-1)]$ as a possible {\opsquare}
(as explained earlier, each such candidate is then tested in constant time). See Figure~\ref{fig:lemma18}. We first establish that the number of such
fragments can be upper bounded by $\Oh(1+|L_{k-2}|)$, and then explain how to generate them in the same time complexity.
\begin{figure}[ht]
\centerline{\includegraphics[scale=.41]{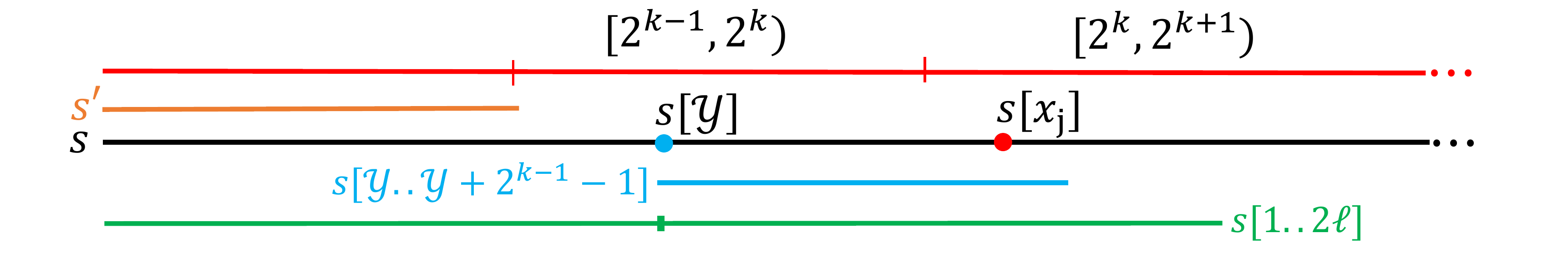}}
\caption{The {\opsquare} $s[1..2\ell]$, the fragment $s[y..y+2^{k-1}-1]$, and $s$ are colored in green, blue, and black, respectively. $s'$ is colored in orange. The red line illustrates the ranges.}
\label{fig:lemma18}
\end{figure}
\begin{lemma}
\label{fewoccurrences}
The number of fragments $s[y..y+2^{k-1}-1]$ such that $x_{j}\in [y,y+2^{k-1}-1]$ and $s'\op  s[y..y+2^{k-1}-1]$ is
upper bounded by $\Oh(1+|L_{k-2}|)$.
\end{lemma}

\begin{proof}
Consider all such fragments $s[y_{1}..y_{1}+2^{k-1}-1], s[y_{2}..y_{2}+2^{k-1}-1], \ldots, s[y_{t}..y_{t}+2^{k-1}-1]$.
Because $x_{j}\in [y_{i},y_{i}+2^{k-1}-1]$ for every $i=1,2,\ldots,t$, either $t=1$ or by the pigeonhole principle there
exists $i$ such that $y_{i+1}-y_{i} < 2^{k-1}/(t-1)$. If $t=1$ then we are done. Otherwise, let $\ell'= |s[y_{i+1}..y_{i}+2^{k-1}-1]|$.
By assumption, $s'\op s[y_{i}..y_{i}+2^{k-1}-1]$ and $s'\op s[y_{i+1}..y_{i+1}+2^{k-1}-1]$, so by the transitivity of $\op $
also $s[y_{i}..y_{i}+2^{k-1}-1]\op s[y_{i+1}..y_{i+1}+2^{k-1}-1]$.
We conclude that $s'[1..\ell']\op s'[y_{i+1}-y_{i}+1..2^{k-1}]$, or in other words $\ell'$ is an op-border of $s'$.
Let $b_{1},b_{2},\ldots,b_{f},b_{f+1}$ be the blocks defined by $\ell'$ in $s'=w[i..i+|s'|-1]$, where each block
is of length $\Delta=2^{k-1}-\ell'$. See Figure~\ref{fig:lemma18blocks}.
Recall that $x_{j}$ is a leftmost occurrence in $w[i..n]$, and by the definition of $y_{i}$ and $y_{i+1}$ we have $x_{j}\in [y_{i+1},y_{i}+2^{k-1}-1]$.
Then, by $s'[1..\ell']\op s'[y_{i}..y_{i}+2^{k-1}-1]$ we obtain that
$x_{j}-y_{i}+1 \in [y_{i+1}-y_{i}+1,2^{k-1}]$ is also a leftmost occurrence in $w[i..n]$.
Hence, we have a leftmost occurrence in $b_{j}$, for some $j\geq 2$.
This allows us to apply Lemma~\ref{leftblocks} and then Lemma~\ref{rightblocks} to
conclude that there is a leftmost occurrence in every block $b_{1},b_{2},\ldots,b_{f}$.
We calculate a lower bound on how many of these leftmost occurrences fall within the range $[2^{k-2},2^{k-1})$:
\begin{align*}
\left\lfloor \frac{2^{k-2}}{\Delta} \right\rfloor - 1  &> \frac{2^{k-2}}{2^{k-1}-\ell'}-2 \\
& = \frac{2^{k-2}}{2^{k-1}-(y_{i}+2^{k-1}-y_{i+1})}-2 = \frac{2^{k-2}}{y_{i+1}-y_{i}}-2 \\
& >  \frac{2^{k-2}}{2^{k-1}/(t-1)}-2 = (t-5)/2.
\end{align*}
For $t<6$, we are done as the number of fragments is $\Oh(1)$. Otherwise, we obtain that $|L_{k-2}| \geq (t-5)/2 \geq t/12$,
thus $t=\Oh(1+|L_{k-2}|)$ always holds as claimed.
\end{proof}
\begin{figure}[ht]
\centerline{\includegraphics[scale=.41]{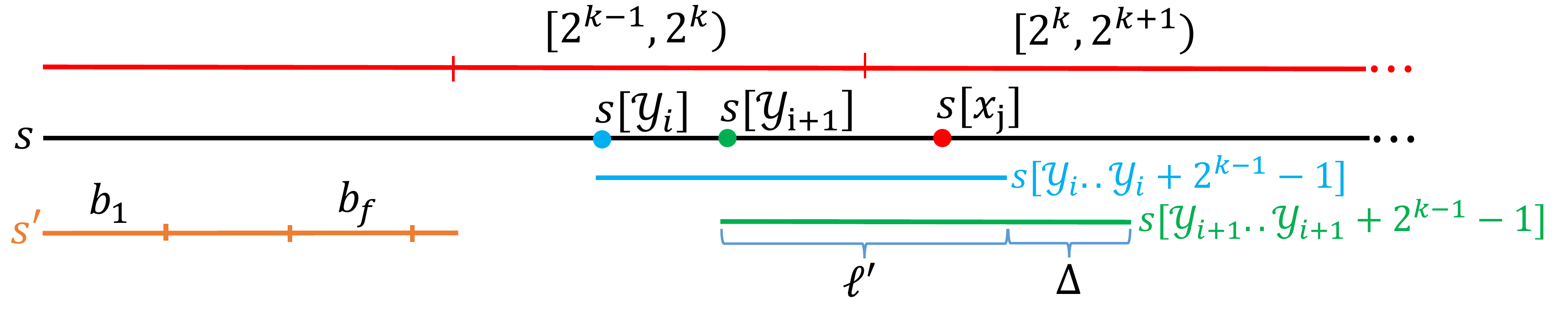}}
\caption{$s$, $s'$,$s[y_{i}..y_{i}+2^{k-1}-1]$, and $s[y_{i+1}..y_{i+1}+2^{k-1}-1]$ are colored in black, orange, blue, and green, respectively. The red line illustrates the ranges.}
\label{fig:lemma18blocks}
\end{figure}
Hence, for every $k$ such that $L_{k}$ is non-empty, we generate $\Oh(1+|L_{k-2}|)$ candidates.
The overall number of candidates generated  by following the above strategy is
$\sum_{k : L_{k} \neq \emptyset} \Oh(1+|L_{k-2}|)=
\Oh(\sigma + \sum_{k} |L_{k-2}|)=\Oh(\sigma)$ as promised. It remains to show how to access all
fragments $s[y..y+2^{k-1}-1]$ such that $x_{j}\in [y,y+2^{k-1}-1]$ and $s'\op  s[y..y+2^{k-1}-1]$ in time
proportional to their number.

We will solve a more general problem, and show how to ensure that, when considering $s=w[i..n]$, for every leftmost occurrence
$x_{j}$ in $s$ we have access to a list of all fragments $s[y..y+2^{k-1}-1]$ such that $x_{j}\in [y,y+2^{k-1}-1]$ and
$s[1..2^{k-1}-1]\op  s[y..y+2^{k-1}-1]$, where $2^{k} \leq x_{j} < 2^{k+1}$. We call this list the
\textit{result} for $i$ and $x_{j}$.

Recall that $s=w[i..n]$, and we consider $i=n,n-1,\ldots,1$ in this order. When we consider $s=w[i..n]$, $i$ becomes a leftmost
occurrence and remains to be so until we reach $s=w[\iprev{i}..n]$ such that $w[i]=w[\iprev{i}]$ (possibly, it is a leftmost occurrence till the very end of the scan).
We can calculate $\iprev{i}$ for every $i$ in $\Oh(\sigma n)$ time by maintaining a list of leftmost occurrences as described earlier.
We say that a position $i$ is $k$-active at position $i'$ when $i'\in [\iprev{i},i]$ and $2^{k} \leq i-i'+1 < 2^{k+1}$.
We observe that, as we consider longer and longer suffixes of $w$, $i$ is first $0$-active, then $1$-active, and so
on until it becomes $k_{i}$-active, and then it is never active again.
Further, indices $i'$ such that $i$ is $k$-active at $i'$ form a contiguous range
$[\ikstart{i}{k},\ikend{i}{k}]$ (the length of each such range is $2^{k}$, except possibly for $k=k_{i}$ when it is shorter). The total
length of these ranges is small as shown below.
\begin{proposition}
\label{smallarrays}
$ \sum_{i,k : k \leq k_{i}} 2^{k} = \Oh(\sigma n) $
\end{proposition}

\begin{proof}
For $k\geq 1$ we can upper bound $2^{k}$ by $2\cdot |[\ikstart{i}{k-1},\ikend{i}{k-1}]|$. Then the sum becomes:
\[ \sum_{i,k : k \leq k_{i}} 2^{k} = n + 2\sum_{i,k : 1 \leq k \leq k_{i}} |[\ikstart{i}{k-1},\ikend{i}{k-1}]|  \leq  n + 2 \sum_{i,k : k \leq k_{i}} |[\ikstart{i}{k},\ikend{i}{k}]| .\]
We observe that every $i'\in [\ikstart{i}{k},\ikend{i}{k}]$ corresponds to $i-i'+1$ being a leftmost occurrence in $w[i'..n]$.
Because there are at most $\sigma$ leftmost occurrences in any $w[i'..n]$, this allows us to upper bound the sum by $\Oh(\sigma n)$.
\end{proof}

This allows us to physically store the results as follows. For every $x$, we have an array indexed by $k\leq k_{i}$, denoted
$R[x]$.
Each entry of this array is an array indexed by $i\in[\ikstart{i}{k},\ikend{i}{k}]$, denoted $R[x][k]$. Finally, each entry of that
array, denoted $R[x][k][i]$, is a pointer
to a list of $y$s such that $x\in [y,y+2^{k-1}-1]$ and $w[i..i+2^{k-1}-1]\op  w[y..y+2^{k-1}-1]$ (note that it is a pointer
to a list and not its separate physical copy). The arrays allow us to access the result for every $x_{j}$, $k\leq k_{x_{j}}$ and $i$
in constant time, by retrieving the pointer $R[x_{j}][k][i]$ (where we first verify that if $i\in[\ikstart{i}{k},\ikend{i}{k}]$).
The total length of all arrays $R[x]$ is only $\Oh(\sigma n)$ by Proposition~\ref{smallarrays}.
Further, the total length of all lists of occurrences that
we need to prepare (assuming that we store every $R[x][i]$ as a pointer to such a list and not their physical copies) is also
$\Oh(\sigma n)$ by the following argument. 
Consider $i$ and $k\leq k_{i}$. Then, we need a list of positions $y$ such that $i\in [y,y+2^{k-1}-1]$ and
$w[y..y+2^{k-1}-1]$ is order-isomorphic to a specific string $s'$. Thus, we can partition all positions $y$
such that  $i\in [y,y+2^{k-1}-1]$ into groups corresponding to order-isomorphic fragments $w[y..y+2^{k-1}-1]$,
and then store a pointer to the appropriate list (possibly null, if there is no $y$). The total number of positions $y$,
over all $i$ and $k\leq k_{i}$, is $\Oh(\sigma n)$ by Proposition~\ref{smallarrays}, which bounds the total length
of all the lists.

It remains to describe how to efficiently calculate the results. This requires partitioning all fragments $w[y..y+2^{k-1}-1]$ such
that $i\in [y,y+2^{k-1}-1]$ and $k\leq k_{i}$ into order-isomorphic groups, and finding for every $i, k\leq k_{i}, i'\in [\ikstart{i}{k},\ikend{i}{k}]$
a pointer to the list of fragments $w[y..y+2^{k-1}-1]$ with $i\in [y,y+2^{k-1}-1]$ that are order-isomorphic to
$w[i'..i'+2^{k-1}-1]$. Both steps can be implemented with the order-preserving suffix tree that is preprocessed in $\Oh(n)$ time
and space for computing a (deterministic) fingerprint of any $\code(w[x..x+2^{\ell}-1])$ in constant time. Here, a fingerprint
is meant as an integer consisting of $\Oh(\log n)$ bits, denoted $\fingerprint_{\ell}(x)$, such that 
$\fingerprint_{\ell}(x) = \fingerprint_{\ell}(x')$ iff $\code(w[x..x+2^{\ell}-1]) = \code(w[x'..x'+2^{\ell}-1])$
(or equivalently $w[x..x+2^{\ell}-1] \op  w[x'..x'+2^{\ell}-1]$). We first describe such a mechanism and then provide
a more detailed description of how to apply it.

\begin{restatable}{lemma}{weightedancestors}
\label{fingerprints}
A compacted trie on $n$ leaves can be preprocessed in $\Oh(n)$ time, so that for any leaf $u$ and integer $k$ we can query
in constant time for a $\Oh(\log n)$-bit fingerprint of the ancestor of $u$ at string depth $2^{k}$.
\end{restatable}

\begin{proof}
This follows by applying the method used to solve the substring fingerprint problem mentioned in~\cite[Lemma 14]{Gawrychowski11}.
Following the description in the full version~\cite[Lemma 12]{Gawrychowski11arxiv}, a compacted trie on $n$ leaves
can be preprocessed in $\Oh(n)$ time so that we can locate the (implicit or explicit) node corresponding to the ancestor
at string depth $2^{k}$ of a given leaf in constant time. If the sought node is implicit (and does not physically exist in the compacted
trie) we retrieve the edge that contains it. Next, if the node is explicit then we return its identifier. If the node is implicit then
we return the identifier of the edge that contains it. Thus, the required range of identifiers is $[2n]$.
\end{proof}

We apply Lemma~\ref{fingerprints} on the order-preserving suffix tree. This allows us to calculate any $\fingerprint_{\ell}(x)$
with the required properties in constant time. Now consider any $i$ and $k\leq k_{i}$. We first compute
$\fingerprint_{k-1}(y)$ for every $y$ such that $i\in [y,y+2^{k-1}-1]$. This takes $\Oh(2^{k})$ time.  Next, we
compute $\fingerprint_{k-1}(i')$ for every $i'\in [\ikstart{i}{k},\ikend{i}{k}]$, also in $\Oh(2^{k})$ time because
$|\ikstart{i}{k},\ikend{i}{k}]|\leq 2^{k-1}$. We sort all fingerprints and partition them into groups corresponding to order-isomorphic
fragments. We need to implement this step in $\Oh(2^{k})$ time as well. To this end, we observe that we need to sort
$\Oh(2^{k})$ integers consisting of $\Oh(\log n)$ bits, which can be done with radix sort in $\Oh(2^{k}+n)$ time.
To avoid paying $\Oh(n)$ for each $i$ and $k\leq k_{i}$, we observe that this is an offline problem, and all sets
corresponding to different $i$ and $k\leq k_{i}$ can be sorted together. In more detail, we sort tuples of the form
$(i,k_{i},\fingerprint_{\ell}(y),y)$ and $(i,k_{i},\fingerprint_{\ell}(i'),i')$. The total number of all tuples is $\Oh(\sigma n)$
by Lemma~\ref{smallarrays} and, as each of them can be treated as an integer consisting of $\Oh(\log n)$ bits,
they can be sorted in $\Oh(\sigma n+n) = \Oh(\sigma n)$ time. Then, we extract the results for each $i$ and $k\leq k_{i}$
from the output. For each $i$ and $k\leq k_{i}$, we consider every group of equal fingerprints. From each group,
we first create a list containing all positions $y$ corresponding to $\fingerprint_{k-1}(y)$ belonging to the group.
Then, for every $\fingerprint_{k-1}(i')$ belonging to the group we store a pointer to this list.
Overall, this takes $\Oh(\sigma n)$ time and allows us to compute all the results in the same time complexity.

\section{Open Problems}
An interesting follow-up to our results is first bounding the number of order-preserving squares that are not order-isomorphic,
and then designing an algorithm that reports all such squares.

\bibliography{biblio}

\end{document}